\documentclass{easychair}

\usepackage{hyperref}
\hypersetup{
  pdfauthor={Wojciech Wawrzyniak},
  pdftitle={A local constant-factor approximation algorithm for MDS problem in anonymous network},
  pdfsubject={A local constant-factor approximation algorithm for MDS problem in anonymous network},
  urlcolor=blue,
}

\usepackage{doc}
\usepackage{makeidx}

\usepackage{graphicx}
\usepackage{epstopdf}
\usepackage{algorithmicx}
\usepackage{algorithm}
\usepackage{algpseudocode}
\usepackage{amsmath}
\usepackage{fixltx2e,amssymb}

\newtheorem{fact}{Fact}
\newtheorem{defi}{Definition}
\newtheorem{lemma}{Lemma}
\newtheorem{theorem}{Theorem}

\algtext*{EndWhile}
\algtext*{EndIf}
\algtext*{EndFor}
\algtext*{EndProcedure}

\begin{document}

%
\title{A local constant-factor approximation algorithm for MDS problem in anonymous network}


\titlerunning{A local constant-factor approximation algorithm for MDS problem in anonymous network}

%
\author{
Wojciech Wawrzyniak\inst{1}\thanks{The research supported by grant N N206 565740.}\\
}

\institute{
Faculty of Mathematics and Computer
Science, Adam Mickiewicz University, Pozna\'n, Poland,
\email{wwawrzy@amu.edu.pl}
 }

\authorrunning{}

\clearpage

\maketitle

\begin{abstract}

In research on distributed local algorithms it is commonly assumed that each vertex has a unique identifier in the entire graph. However, it turns out that  in case of certain classes of graphs (for example not lift-closed bounded degree graphs) identifiers are unnecessary and only a port ordering is needed \cite{GHS_12a}. One of the open issues was whether identifiers are essential in planar graphs. In this paper, we answer this question and we propose an algorithm which returns constant approximation of the MDS problem in $\mathcal{CONGEST}$ model. The algorithm doesn't use any additional information about the structure of the graph and the nodes don't have unique identifiers.
We hope that this paper will be very helpful as a hint for further comparisons of the  unique identifier model and the model with only a port numbering
in other classes of graphs.

\end{abstract}

\setcounter{tocdepth}{2}

%
%

\vfill \eject
\section{Introduction}
	In recent years, there has been a growing interest in designing distributed local algorithms. 
It might come out of the easiness of applying these algorithms in reality. They run very fast (in constant time) and are tolerant to the network structure changes and node failures. 
It turns out the running time of these algorithms is completely decoupled from the size of the network and each node takes its decision based only on the knowledge about its k-neighbourhoods.   
This fact is very important for the scalability of an algorithm in large networks. If the structure of the network changes (i. e. a vertex is removed), then an algorithm must be re-called to repair a solution only for a small surrounding of the removed vertex. It is a significantly faster solution than in case of standard algorithms requirements, which require re-execution of the algorithm on the entire network. \par
	In some research on designing local algorithms(but not strictly local), it is allowed that nodes have a knowledge about the $f(n)$-neighbourhood, where $f(n)$ is a function that depends on the number of vertices in the network. However, in this paper we only consider strictly local algorithms, that do not need any additional information about the structure of the graph and don't have unique identifiers, so they satisfy much stronger assumptions. \par
	In recent years, several deterministic distributed local algorithms have been proposed. They return solutions that are good approximations of various problems (e.g. minimum edge cover, minimal dominating set\cite{LOW_10tik}, semi-matching\cite{CHKSW11, CHSW12a}), in constant time in different classes of graphs (e. i. bounded degree graphs, planar graphs). However, these algorithms very often assume that nodes have unique identifiers. This assumption could be very important if we consider a more ''real'' model, in which in a single communication round, each vertex can send a message which contains  at most $O(\log{n})$ bits, where $n = |V(G)|$ is the number of vertices in the graph. This limitation makes it impossible to e.g. detect small cycles in the network, gather knowledge of 2-hop neighbourhoods. 
	Recently in a paper \cite{GHS_12a} the authors G\"{o}\"{o}s et al. 
 have shown that for lift-closed bounded degree graphs, a model with unique identifiers (known as $\mathcal{LOCAL}$ \cite{L_92a}) and model with a port numbering only(known as PO model\cite{GHS_12a}), are practically equivalent. However, techniques used in their work do not allow us to consider the equivalence of these models for Minimum Dominating Set($MDS$) problem in planar graphs. It is known\cite{LOW_10tik} that there exists an algorithm for planar graphs which, in constant time, returns a constant approximation of the $MDS$ in  model with unique identifiers and an unbounded message size. \par
	It turns out that there also exists a strictly local algorithm for planar graphs, that in the model without unique identifiers and with upper bounded message size, finds constant approximation of the Minimum Dominating Set. 
\subsection{Related Work}\parskip 0pt 
A distributed algorithm is called a local algorithm if it completes in a constant number of synchronised communication rounds. If we assume that the nodes do not have any additional information about the other vertices, then we say that our algorithm is {\it strictly local}. \par 
The research on local algorithms has been ongoing for several years (\cite{ABI_86,CV_86,II_86,L_92a,Luby_86,NS_95,Peleg_00}), but the strictly local algorithms gained the increased interest just recently. There are now more than one hundred works referring, more or less closely, to the topic of such algorithms. Thus, it is not possible to briefly describe all of these publications. The best way to study this topic is to read excellent survey\cite{Suomela_13} written by Suomela. That article describes all the important results obtained so far by all the researchers. One of many open questions is an issue raised in a paper \cite{GHS_12a} concerning the similarity of two models: a model with unique identifiers and a model with only a port numbering for MDS problem in planar graphs. We answer this question. \par
	One of the first papers, that considered network without unique identifiers, has been written by Angluin \cite{Angluin_80}. Unfortunately, in 1992, Linial showed in \cite{L_92a} that there is no algorithm that, in constant time, finds a Maximal Independent Set in a cycle in the unique identifiers model. This result shows how difficult it is to find a fast distributed algorithm and it is even more difficult if we consider {\it strictly local} algorithms only. Thankfully, in 1995 Naor and Stockmeyer in \cite{NS_95}  introduced the concept of Local Checkable Labelling(${LCL}$) problems and showed that if there is a local algorithm in a model with unique identifiers on nodes then there is also order-invariant local algorithm which uses only the fact that for each pair $v,u$ of vertices $id(v)<id(u)$ or $id(v) > id(u)$. So from the point of view of the ${LCL}$ problems both models are almost equivalent. Note that the class of ${LCL}$ problems contains among others the maximal independent set or vertex colouring. Thus, a natural question then came up, whether there exists an algorithm which, without information about the sequence of vertices is able to solve any non-trivial problem. 
	Kuhn and Wattenhofer in \cite{KW_05}, presented the first local but randomized algorithm for bounded degree graphs. Their algorithm does not require long messages. Then in \cite{KMW_06} the algorithm has been improved by Kuhn et al. 
 Notice that both approaches used the method of linear programming. The first local algorithm for MDS problem for planar graphs was proposed by Lenzen et al. 
in \cite{LOW_10tik}, but their algorithm requires long messages and unique IDs on nodes. \par
	There is also a lower bound for possible approximation factor of an algorithm.  In \cite{CHW_08a} 
it has been shown that there is no algorithm which in a constant number of communication rounds returns an $(5-\epsilon)$ approximation of the MDS in planar graphs.

\subsection{Main Results and Organisation}\parskip 0pt 
Our main result is summarised in the following theorem. Let $M$ denote an arbitrary MDS in a planar graph $G=(V,E)$. 
\begin{theorem}\label{thm:main}
	Let $G=(V,E)$ be a planar graph and $D$ be a set returned by algorithm {\it PortNumberingMds}. Then 
$|D| \leq O(|M|).$
\end{theorem}
The rest of this paper is structured as follows. We begin by describing the computational model and notation used in this paper. Then in section \ref{sec:Algorithm} we briefly introduce the principle of our algorithm and its formal pseudocode. Next, in section \ref{sec:Analyse}, we present the analysis of the correctness of our algorithm, and compute the approximation factor of the algorithm. At the end, in section \ref{sec:Conclusion}, we summarise our considerations.

\subsection{Model and Notation}\parskip 0pt 
In this paper we work in a synchronous communication model and as a representation of the network we use a~planar graph $G=(V,E)$. Edges in the graph will correspond to communication links and processors will correspond to vertices from the set $V$.
Moreover, we assume that each vertex has its own labelling of its incident edges and vertices do not have unique identifiers and also do not have any additional information. \par
	In order to facilitate the reader to understand this paper, we use the same notations as in \cite{LOW_10tik}. For nodes $A \subseteq V$ we define the set of inclusive neighbourhood of $A$ as $N^+_A:=\left\{ v: v\in A \vee \exists_{e=uv \in E} u\in A \right\}$. We also denote the neighbours of $A$ not in $A$ as $N_A := N^+_A \setminus A$. To simplify the notation in cases where $A=\left\{a\right\}$ we may omit the braces, e.g. $N_a$ instead of $N_{\{a\}}$.
	
\section{Constant approximation in $\mathcal{CONGEST}$ model}\parskip 0pt 
\subsection{Algorithm}\label{sec:Algorithm}\parskip 0pt 
	The key idea of the algorithm is based on an appropriate use of planarity of the graph $G$. Intuitively, some vertex $v$ should belong to the dominating set $D$ if it dominates a lot of its neighbours. However, in reality, such approach does not give a constant approximation as we can see in the Figure \ref{fig:example_bad_approximation}.  
	This situation occurs if graph $G$ contains many vertices with big common neighbourhood. In our algorithm we first dominate only a small subset of these vertices (step \ref{alg_D1} and \ref{alg_D2} of the algorithm). So we avoid unnecessary adding of multiple vertices which dominate the same or almost the same neighbourhoods. \par
{
\begin{algorithm}
\caption{PortNumberingMds}\label{Alg_MinimalDS}
\begin{algorithmic}[1]
\State $D := \emptyset$.
\State $D_1 :=$ {\it Hop2Dominate(G,D)}, \ \ $D := D \cup D_1$ \label{alg_D1}
\State $D_2 :=$ {\it Hop2Dominate(G,D)}, \ \ $D := D \cup D_2$ \label{alg_D2}

\For {$v\ \in V$ \textrm{ in parallel}}
	\State $\delta^{V \setminus N^+_D}_v := |N^+_v \setminus N^+_D|$ 
	\If {$v \notin N^+_D$}
		\State $\mu_v := \max_{w \in (N^+_v \cap N_D)}{\{\delta^{V \setminus N^+_D}_w\}}$
		\State choose any $w(v) \in \{w \in (N^+_v \cap N_D) : \delta^{V \setminus N^+_D}_w = \mu_v\}$ \label{alg_w(v)}
		\State $D_3:= \{ w(v): v \notin N^+_D \}$, $D:=D \cup D_3$  \label{alg_D3}
	\EndIf
\EndFor
\State \Return{ $D$}
\end{algorithmic}
\end{algorithm}
}
\vspace{-15pt}
{
\begin{algorithm}\floatname{algorithm}{Function}
\caption{Hop2Dominate(G,D)}\label{Alg_2hop_dominate}
\begin{algorithmic}[1]
\For {$v \in V$ \textrm{ in parallel}}
	 $\delta^{V \setminus D}_v := |N^+_v \setminus N^+_D|$ 
\EndFor
\For {$v \in V \setminus D$ \textrm{ in parallel}}
	\State $\Delta^{V \setminus D}_v := \max_{w \in N^+_v}{\{\delta_w^{V \setminus D}\}}$ 
	\State choose any $x(v) \in \{u \in N^+_v : \delta^{V \setminus D}(u) = \Delta^{V \setminus D}_v \}$ \label{alg_x(v)}
	\State $X := X \cup \{ x(v) \}$ \label{alg_X}
\EndFor

\For {$v \in V$ \textrm{ in parallel}}
	\State $\delta^X_v := |N^+_v \cap X|$ 
	\If {$v \in X$}
		\State $\xi_v := \max_{w \in N^+_v}{\{\delta^X_w\}}$ 
		\State choose any $d(v) \in \{w \in N^+_v : \delta^X_w = \xi_v\}$ \label{alg_d(v)}
		\State $D_{new}:= \{d(v): v\in X\}$
	\EndIf
\EndFor
\State \Return {$D_{new}$}
\end{algorithmic}
\end{algorithm}
}

	In the next round each vertex not dominated yet adds to the set $D$ a dominated vertex with biggest residual degree from its dominated neighbourhood. The planarity of the graph $G$ ensures that there is a small number of such added vertices. 
	To prove that both sets are small, we will use well known fact that Jordan curve divides the plane into two regions - an {\it interior} and an {\it exterior}, so that any cycle in a planar graph $G$ divides the graph into two parts without edges between their interiors. We partition our plane graph into disjoint regions in such way that the number of regions are proportional to the size of the set $D$ and moreover, in each region there is at least one vertex from the set $M$.

\parskip 0pt 
\subsection{Analysis}\label{sec:Analyse}\parskip 0pt 
	As can be easily seen, the algorithm can be performed in a constant number of communication rounds and returns a dominating set due to last round (step \ref{alg_D3}), where all not dominated vertices add exactly one of their  neighbours to the dominating set $D$. Therefore, in our analysis we only need to show that the numbers of vertices added to the dominating set $D$ in steps \ref{alg_D1}, \ref{alg_D2} and \ref{alg_D3} are small enough that our algorithm returns solutions which are a constant approximation of an optimal MDS. 
	To simplify notation in our analysis, we assume that the set of vertices added in step \ref{alg_D1}, \ref{alg_D2} and \ref{alg_D3} will be denoted by $D_1$, $D_2$ and $D_3$ respectively, and some fixed optimal solution will be denoted as $M$. We need to recall the following well-known lemma.
\begin{lemma}\label{lem_planar_graph}
A minor of a planar graph is planar. A planar graph of $n$ nodes has less than $3n$ edges. A planar bipartite graph of $n$ nodes has less than $2n$ edges.
\end{lemma}

	We will begin the analysis of our algorithm with estimating the maximal number of vertices added to the set $D_1 \setminus M$. To bound this value we need to define a special subgraph $G_1$ of graph $G$. 

\begin{defi}\label{definition_arrow_graph}
	Let graph  $G_1=(V_{1},E_{1})$ be a subgraph of $G=(V_G,E_G)$ constructed in the 
following way:
{
\vspace{-5pt}
\begin{itemize}\itemsep=-2pt \topsep=-4pt
\item[i)] $V_{1} := X \cup D_1$ and $E_{1} = \emptyset$, where $X$ is a set from step \ref{alg_D1} of the algorithm.
\item[ii)] Add all edges between vertices from $V_{1}$.  
\item[iii)] Add minimal number of edges (from $E_G$) and nodes (from $V_G$) such that each vertex from the current set $V_1$ has adjacent vertex from the set $M$ or is contained in the set $M$ (so for each $v\in V_1$ we have $N^+_v(G_1) \cap M \neq \emptyset$).  
\end{itemize}
}
\end{defi}
\noindent In order to simplify the description of proofs, we will also introduce the following notation (see Figure \ref{fig:Graph_G1}):
\begin{spreadlines}{0pt}
\begin{align*}
\setlength{\belowdisplayskip}{0pt}%
\setlength{\abovedisplayskip}{0pt}%
X_M&:= \left\{ v:v \in (X \cap M) \right\}, & &Y_M:= \left\{ d(v):v \in X_M \right\}  ,\\ 
X_S&:= \{ v:v \in X\setminus M \wedge |N^+_{d(v)} \cap X| \leq c   \}, & &Y_S:= \left\{ d(v):v \in X_S \right\}  ,\\ 
X_L&:= \left\{ v: v \in X \setminus (X_M \cup X_S)  \right\}, & &Y_L:= \left\{ d(v):v \in X_L, \right\}, \\
E_i&:=\left\{ \{v,x\}: v\in (Y_L \setminus M) \wedge x \in X_i \right\}, & & i \in \left\{M,S,L\right\}, \\
Y&:= Y_M \cup Y_S \cup Y_L
\end{align*}
\end{spreadlines} 
where $d(v)$ is a vertex chosen in the step \ref{alg_d(v)} of the algorithm.  			
	Notice that not all of the subsets are disjoint, for example, it is possible that some fixed vertex $v$ belongs to both sets $Y_L$ and $X_L$ ($v\in Y_L \cap X_L$). 
\begin{figure}[ht!]
	\centering
	\includegraphics[width=0.55\textwidth]{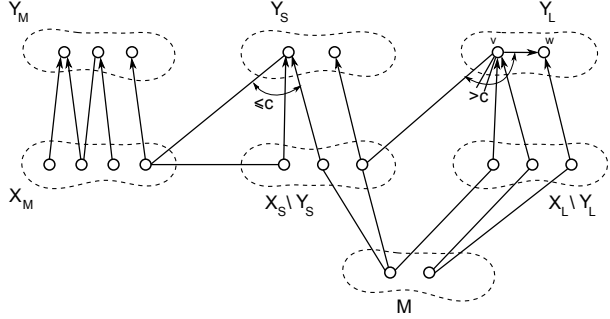}
	\caption{An example of the graph $G_1$. If $d(v)=u$ then edge $e=vu$ is marked by an arrowhead. }
	\label{fig:Graph_G1}
\end{figure}
	To show that the maximal number of vertices in the set $D_1 \setminus M \subseteq Y_M \cup Y_S \cup Y_L$ is comparable to the order of the set $M$, we will consider the size of each set $Y_M$, $Y_S$, and $Y_L$ separately.
	This analysis is contained in Lemma \ref{lem_y1}, Lemma \ref{lem_y2} and Lemma \ref{lem_y3}. \par

At the beginning we will prove, a simple but very useful fact.
	
\begin{fact}\label{fac_y_less_than_x}
	$ |Y_i| \leq |X_i| \textrm{\ \ for each \ \ } i \in \left\{M,S,L\right\} $
\end{fact}
\begin{proof}
	Note that the vertex from the set $Y_i$ has been added in step \ref{alg_D1} of the algorithm by one of the vertices in $X_i$. In addition, each vertex $x \in X_i$  adds at most one vertex to $D_1$. Thus, the order of the set $Y_i$ cannot be greater than the order of the set $X_i$.
\end{proof}

\begin{lemma}\label{lem_y1}
$|X_M| \leq |M |\textrm{\ \  and \ \ } |Y_M| \leq |M|.$
\end{lemma}
\begin{proof}
	The set $X_M$ is a set contains the elements which both belong to sets $X$ and $M$. Hence the order of $X_M$ is less or equal to the order of $M$ ($|X_M| \leq |M|$). Moreover, using Fact \ref{fac_y_less_than_x}, we obtain that $|Y_M| \leq |X_M| \leq |M|$.
\end{proof}

\begin{lemma}\label{lem_y2}
$|X_S| \leq c|M| \textrm{\ \  and \ \ } |Y_S| \leq c|M|.$
\end{lemma}
\begin{proof}
	In the step \ref{alg_D1} of the algorithm every vertex $v \in X$ adds its adjacent vertex $w \in N^+_v$ with the biggest residual $X$ degree $\delta^X_w$ (where $\delta^X_w := |N^+_w \cap X|$) from the inclusive neighbourhood. The definition of the set $X_S$ implies that every vertex $v \in X_S$ that does not belong to $M$ has at least one neighbour in the set $M$, so that vertices from the set $X_S$ have to be dominated in the optimal solution $M$. Let us define a set $A:=N_{X_S} \cap M$ then for all $m \in A$ we have that the residual $X$ degree of $m$ is less or equal to $c$ ($\delta^X_m \leq c$). Otherwise, the vertex $v$ would not belong to $X_S$ ($v \notin X_S$) because its residual $X$ degree is bigger than $c$. Hence, every vertex $m \in A \subseteq M$ dominates at most $c$ vertices from the set $X_S$ so $|M| \geq |X_S|/c$. Using $|Y_S| \leq |X_S|$ from Fact \ref{fac_y_less_than_x}, we obtain that	$|M| \geq |X_S|/c\geq |Y_S|/c$.
\end{proof}

Our goal is to show that $|D_1 \setminus M|=O(| M |)$ so it is left to prove that the maximal number of vertices in $Y_L$ is small ($|Y_L| = O(|M|)$). For this purpose, we will use a technique of splitting the graph $G$ into bunches and then we will show that each induced region of a bunch contains many vertices from the set $M$.
We start by defining what we mean by a term {\it bunch}, which was first introduced in \cite{CHW_08b}.
\begin{defi}\label{def:bunch}
Let $G = (V,E)$ be a planar graph, $S \subseteq V$, $T \subseteq V$ and $W \subseteq V$. A $v_i$-$v_j$-path is called {\bf S-T-W-special} if it has the form $v_iuv_j$, where $v_i \in S$,  $u \in T$ and $v_j \in W$. 
\end{defi}

Although our algorithm works in planar graphs, in the analysis we assume that the given graph $G$ is plane. Let us recall some basic
theoretical graph terminology for planar graphs.
If $G$ is a plane graph in  $R^2$ then maximal open set $f$ in $R^2 \setminus G$  such that any two points in $f$ 
can be connected by a curve contained in $f$ is called a
face of $G$. Let $P,Q$ be two special $v_i$-$v_j$-paths. In
any plane drawing, graph $P \cup Q$ contains exactly one
bounded face. (We will assume here that the face is
empty if $P = Q$.) Now we set $F(P \cup Q) := f$ and
$Reg[P \cup Q] := (P \cup Q) \cup f$ where $f$ is the bounded
face in the drawing of $P \cup Q$. 

\begin{defi}\label{def:RegionR}
Let $G = (V,E)$ be a plane graph and let $v_i \in S, v_j \in W$, $T \subset V$ where $i \neq j$. A maximal set $B$ of {\it S-T-W-special} paths between $v_i$ and $v_j$ is called a {\bf S-T-W-bunch between $v_i$ and $v_j$} if there exist two distinct paths $P,Q \in B$ such that all paths from $B$ are contained in $Reg[P \cup Q]$ and no vertex from $S \cup W$ is contained
in $F(P \cup Q)$. In addition, the paths $P,Q$ will
be called {\bf the boundary paths} of $B$. Moreover if a bunch $B$ contains at least five special paths then we say that $B$ is a {\bf large bunch}.
\end{defi}

To simplify the notation, if the sets $A$, $B$, $C$ are clear from the context, we will write  {\it special paths} instead of {\it A-B-C-special paths}.
In one of the last lemmas in this paper we will consider special paths and bunches of length three. Their definition is analogous to the definitions of bunches with special paths of length two.

After defining the concepts of bunches and special paths, next, in Fact \ref{lem_bunches_amount} and Lemma \ref{lem_y3}, we will estimate their sizes. Then, in Lemma \ref{lem_many_m_in_big_bunch}, we will show that most of regions designated by the bunches contain many vertices from the set $M$.	The proof of Fact \ref{lem_bunches_amount} is quite complicated and at the beginning we show that the number of connected components of the induced subgraph is smaller than $|M|$.

\begin{lemma}\label{lem_few_components}
	Let $G=(V,E)$ be a graph and $M$ be a dominating set in $G$. If $H=(V_H, E_H)$ is a subgraph of $G$ such that $M \subseteq V_H$ and every vertex $v \in V_H \setminus M$ contains at least one adjacent vertex from $M$ then 
$ \omega(H) \leq |M|.$
\end{lemma}
\begin{proof}
	Let $Z_1, Z_2, \ldots , Z_k$ be a partition of $V_H$ to minimal number of  connected components. If a set $Z_i$ contains at least one vertex $v \in V_H \setminus M$ then there is a vertex $m \in M$ such that $\{m,v\} \in E_H$. Hence each connected component $Z_i$ contains at least one vertex from $M$. In other case there is no vertex $v \in V_H \setminus M$ in component $Z_i$ then since $Z_i \neq \emptyset$ thus $Z_i$ contains at least one vertex from a set $M$.  
	 We obtain that each connected component contains at least one vertex from $M$ thus $\omega(H) \leq |M|$.
\end{proof}

\begin{lemma}\label{lem_bound_number_of_bunches}
	Let $G=(V,E)$ be a planar graph and $A,B,C \subseteq V$ be subsets of vertices such that the sets are pairwise disjoint and each vertex from $B$ is adjacent to at least one vertex from each sets $A$ and $C$. Then graph $G$ contains at most $4(|A| + |C|) + \omega(V)$ $A$-$B$-$C$-bunches, where $\omega(V)$ denote the number of connected components in graph $G$.
\end{lemma}
\begin{proof}
	To bound the number of bunches in the graph $G$ more effort is required. First of all, we need to define a multigraph $H=(V_{H}, E_{H})$ obtained from $G$ by contracting each vertex $x \in B$ to any adjacent vertex $m \in C$ and adding edge between contracted vertices and neighbours of a vertex $x$ from a set $A$ (see Figure \ref{fig:Graph_H}). Let vertices $u,w \in V$  was contracted in the graph $H$ then we say that path $vuw$ from the graph $G$ ($v \in A$, $u \in B$, $w \in C$) corresponds to edge $e=\{v, uw\}$ in the graph $H$.  Notice that each vertex $x \in B$ is adjacent with exactly one vertex $m \in C$. \par
\begin{figure}[ht!]
	\centering
	\includegraphics[width=0.65\textwidth]{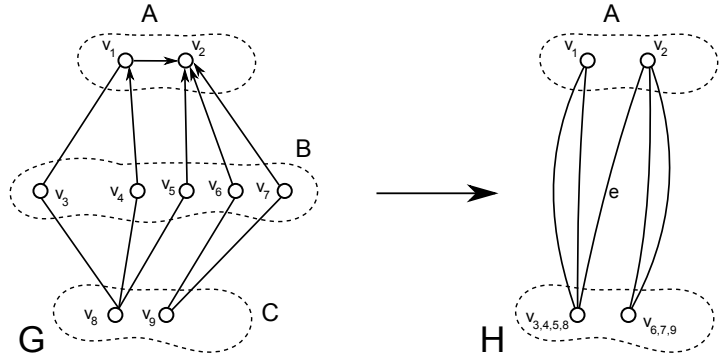}
	\caption{An example of the construction of the graph $H$. We say that path $v_2v_5v_{8}$ from the graph $G$ corresponds to edge $e$ in the graph $H$.} 
	\label{fig:Graph_H}
\end{figure}

	Let us consider a connected component of the multigraph $H$ ($H[Z_i]$), where  $Z_i \subseteq V_{H}$ denotes the set of all vertices from such component.  Then we can find spanning tree $T:=T_{H}^{Z_i}$ in a multigraph $H[Z_i]$. From a well known Lemma \ref{lem_planar_graph} we know that a multigraph $H[Z_i]$ is planar. \par
	Consider a plane drawing of $H[Z_i]$. Let for every
vertex $v \in Z_i$ and $\epsilon_v > 0$ define a ball $C_v$
around a vertex $v$ of radius $\epsilon_v$, such that $C_v$ intersects only with these edges of $H[Z_i]$ that contain $v$ and does not contain points from
other balls. We denote a connected region of $C_v \setminus T \subseteq R^2$
as a {\it side} of vertex $v$. It is obvious that every edge from $E(H[Z_i]) \setminus E(T)$ that
contains $v$ reaches $v$ by some side $s$ of a vertex $v$. In this case
we will say the edge ends in side $s$ (see Figure \ref{fig:sides}).

\begin{figure}[ht!]
	\centering
	\includegraphics[width=0.35\textwidth]{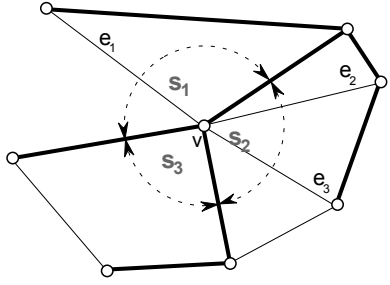}
	\caption{An example of sides in an arbitrary graph, as we can see that edges $e_2$, $e_3$ end in side $s_2$ of vertex $v$. The edges of the tree $T$ are shown in bold. }
	\label{fig:sides}
\end{figure}

We also have to prove similar fact like in paper \cite{CHW_08b}.

\begin{fact}\label{fac_two_edges_the_same_side}
The multigraph $H$ contains at most two edges $e, e'$ of $E(H[Z_i]) \setminus E(T)$ such that $e$ and $e'$ end in the same sides and corresponding special paths of edges $e,e'$ in $G$ belong to different {\it A-B-C-bunches} in corresponding graph $G[Z'_i]$, where $Z'_i$ denotes all contracted vertices in $Z_i$. Furthermore there is at most one such pair of edges $e,e'$ in a multigraph $H$.
\end{fact}
\begin{proof}
Let F be the set of $u$-$v$ edges from $E(H[Z_i]) \setminus E(T)$ that end in the same sides
of $u$ and $v$. Assume that $e,e' \in F$ belong to different bunches then 
$C_1 := uTv + e$ is a cycle and
consequently every other $u$-$v$ edge i.e. $e'$ must be contained
in one of the regions of $C_1$. Because corresponding special paths of $e$ and $e'$ are contained in different bunches in $G[Z'_i]$ thus the region $R[C_1] \cup R[C_2]$ where $C_2:=uTv+e'$ contains all vertices from $Z'_i$. 
If there is other $u$-$v$ edge $e''$ which belongs to different bunch than $e$ and $e'$
contained in the bounded face of $C_1$ or bounded face $C_2$ then there is a vertex $z$ from the set 
$Z_i$ which is contained in the
bounded region of the cycle $ueve'u$. Then e,e' and e'' end in different
side of u (contradiction).  Moreover if graph $H'$ contains such edges $e,e'$ then from planarity there is no any other pair of edges $e_2, e'_2 \in  E(H[Z_i]) \setminus E(T)$ which ends in the same sides of two vertices.
\end{proof}

\begin{figure}[ht!]
	\centering
	\includegraphics[width=0.76\textwidth]{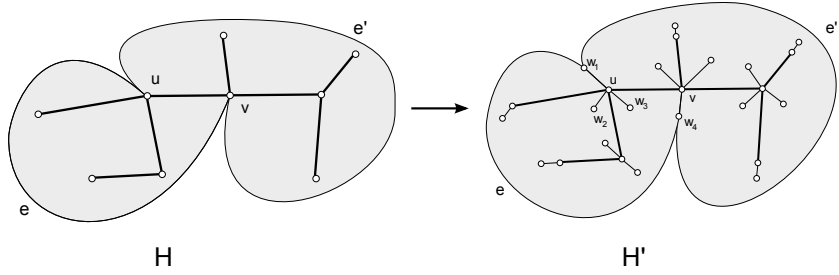}
	\caption{An example that sum of two regions $R[C_1] \cup R[C_2]$ contain all vertices from $Z_i$. The edges of the tree $T$ are shown in bold.}
	\label{fig:tree_cycle_region}
\end{figure}

Let $H'[Z_i]$ be the supergraph of $T$ obtained as follows.
For every vertex $v \in V_T$ put a vertex $w_v$ in each
side of $v$ and join it with $v$ by one edge. The set of new added vertices we denote as $V_{T'}$. Substitute
the edge from $E(H[Z_i]) \setminus E(T)$ which ends in the side
of $v$ containing $w_v$ with the edge that ends in $w_v$. Let $e_1,e_2,\ldots ,e_k \in E(H[Z_i]) \setminus E(T)$ be a maximal set of edges which corresponds to special paths in some fixed bunch from $G[Z'_i]$ then we remove edges $e_2,e_3, \ldots, e_k$ from $H'[Z_i]$.
The supergraph $H'[Z_i]$ is a planar multigraph and using Fact \ref{fac_two_edges_the_same_side} we obtain that almost every pair of vertices (except for one) could be connected by at most one edge from a set $E(H'[Z_i]) \setminus E(T)$. Let us notice that for each bunch in $G[Z'_i]$ there exists disjoint corresponding edge in $H'[Z_i]$.
	For every vertex $v \in T$ we add exactly $deg_T(v)$ new vertices, thus we can simply determine the number of new added vertices from a supergraph $H'[Z_i]$
\begin{align}\label{H1prim_vertices}
|V_{T'}| = \sum_{v\in V_T} d_T (v) = 2|T| - 2.
\end{align}
Let us observe that in our lemma we consider $A$-$B$-$C$-bunches, where sets $A$, $B$, $C$ are pairwise disjoint. Thus each special path of considered bunches has one endpoint in set $A$ and one in $B$. Hence we may assume that our supergraph is bipartite. Using Lemma \ref{lem_planar_graph}, Fact \ref{fac_two_edges_the_same_side}  and equation (\ref{H1prim_vertices}) we obtain that number of edges
\begin{align}
||H'[Z_i]|| \leq 2|V_{T'}| + |T'| + (|T|-1) + 1 \leq 7|H[Z_i]| - 6.
\end{align}
Notice that edges between vertices from a spaning tree $T$ and new vertices $V_{T'}$ was added in supergraph $H'$ but not exists in $T$ and moreover some edges (i.e. $w_3w_4$) belong to the same bunch. We can omit such edges in our calculation, thus the maximal number of bunches in the graph $G[Z_i]$ is less than 
$4|Z_1|$.	Unfortunately, the graph $G$ may not be connected, therefore the number of bunches $\mathcal{B}_1$ may be greater than $\sum_{i} 4|Z_i|$ due to some bunch $B$ could be contained in a region of other bunch $B'$. If we consider creating a multigraph $H$ by sequentially adding connected 
components then in $i$-th step after adding corresponding $G[Z_i]$ component we create at most $4|Z_i|+1$ new bunches. 
So a graph $G$ contains at most $4(|A|+ |C|) + \omega(G)$ bunches.  \par
\end{proof}
\begin{fact}\label{lem_bunches_amount}
	Let $A:=Y_L \setminus M$, $B:=X_L \setminus Y_L$ and $C:=M$.  Then the graph $G_1$ contains at most $4|Y_L \setminus M|+ 5|M|$ A-B-C-bunches. This set of bunches we denote by $\mathcal{B}_1$.
\end{fact}
\begin{proof}
	This follows directly from Lemma \ref{lem_few_components} and Lemma \ref{lem_bound_number_of_bunches}.
\end{proof}

\begin{lemma}\label{lem_many_m_in_big_bunch}
	Let $B \in \mathcal{B}_1$ be a bunch such that $B$ contains at least five {\it $(Y_L \setminus M)$-$(X_L \setminus Y_L)$-$M$-special paths} in the graph $G_1$ ($b_B \geq 5$). Then
	$|M \in F(B)| \geq \left\lceil \frac{b_B-3}{2}\right\rceil$ , 
where $M \in F(B) := \{m \in M : m \in F(P,Q)$ and $P$, $Q$ are boundary paths of a bunch $B \}$.
\end{lemma}

\begin{proof}
 Let us consider the structure of a subgraph of $G$ induced by vertices contained in a region designated by a boundary of special paths of some bunch $B \in \mathcal{B}_1$ ($R[B]$). Recall that we denote a number of special paths in a bunch $B \in \mathcal{B}_1$ as $b_B$ and we take into account only bunches $B \in \mathcal{B}_1$ such that $b_B \geq 5$. Hence each considered bunch contains a vertex $v \in Y_L \setminus M$, a vertex $m \in M$ and at least five vertices from the set $X_L \setminus Y_L$ (see Figure \ref{fig:figure_some_bunch}). 
	 Moreover, a bunch $B$ creates at least $b_B-1$ disjoint regions in the graph $G \setminus B$. We will show that many of them contain vertices from $M$ and, more precisely, each region $R[B]$ contains at least $\left\lceil (b_B-3)/2 \right\rceil$ vertices from $M$. \par
\begin{figure}[h!]
	\centering
	\includegraphics[width=0.42\textwidth]{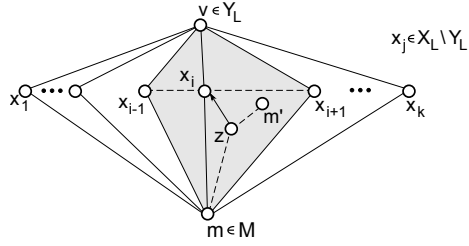}
	\caption{Example of a subgraph of $G$ for some bunch $B \in \mathcal{B}_1$. Region $F(v, x_{i-1}$, $m$, $x_{i+1})$ is marked with grey colour.}
	\label{fig:figure_some_bunch}
\end{figure}
Since vertex $x_i$ belongs to the set $X_L \setminus Y_L$, thus $x_i$ was added to  $X$ by some vertex $u \in V$ 
in the step \ref{alg_X} of the algorithm, as a vertex with the largest degree in the neighbourhood $N^+_u(G)$. It is possible that $v=u$ but note that a vertex $v$ can add only one such vertex. Let us assume that $u \neq v$. Using an assumption that $b_B \geq 5$ we obtain $deg_G(m), deg_G(v) \geq 5$ thus an interior vertex $x_i$(see Figure \ref{fig:figure_some_bunch}) could not have been added by any of the vertices $x_{i-1}$, $x_{i+1}$ or $m$ until some other node $z$ adjacent to $x_i$ exists in $F(v, x_{i-1}$, $m$, $x_{i+1},v)$(see Figure \ref{fig:figure_some_bunch}). Hence each  interior vertex $x_i \in F(B)$ is adjacent with at least one vertex $z$ from region $F(v,x_{i-1},m,x_{i+1},v)$ such that at least one of the following cases is satisfied or $x(v) = x_i$
\vspace{-5pt} 
\begin{enumerate}\itemsep=-2pt \topsep=-4pt 
 	\item[a)] $z \in M$ 
	\item[b)] $\exists m' \in M$ such that $\{m',z\} \in E_G$ and $m' \in F(v,x_{i-1},m,x_{i+1},v)$
\end{enumerate}
Let $z_1, z_2, \ldots , z_k$ be a set of vertices lying inside $F(v,x_{i-1},m,x_{i+1},v)$ and adjacent to a vertex $x_i$. Suppose that $x(v)\neq x_i$ and case a) is not satisfied for any $z_j$, so $x(v) \neq x_i$ and $z_1, z_2, \ldots , z_k \notin M$. In the optimal solution $M$ every vertex $v \in V$ belongs to $M$ or has a neighbour in this set, thus there exist vertices $m'_1, m'_2, \ldots, m'_k \in M$ such that each $m'_j$ dominates $z_j$ ($j \in \{1,2,\ldots,k \}$). Recall that there exists $z_l \in \{z_1, z_2, \ldots , z_k\}$ such that $x(z_l) = x_i$ (determined in step \ref{alg_x(v)} of 	the algorithm) so $deg_G(m'_l) \leq deg_G(x_i)$. Assume by contradiction, that case b) is also not satisfied for each $m'_1, m'_2, \ldots, m'_k$. Then $m'_1= m'_2= \ldots = m'_k = m$, but in this case $deg_G(m) > deg_G(x_i)$ and thus there is no vertex $z_l$ such that $x(z_l)=x_i$. It is a contradiction with assumption that $x_i \in X$. Hence at least one of the cases a), b) is satisfied.\par 
	In a subgraph induced by {\it boundary paths} of a bunch $B$ there are exactly $b_B-2$ internal vertices from the set $X_L \setminus Y_L$ and furthermore at most one such vertex could be chosen by vertex $v \in Y_L \setminus M$ from this bunch. So at least $b_B-3$ internal vertices of the bunch have corresponding vertex $m'\in M$ which is contained in the region $F(v,x_{i-1},m,x_{i+1},v)$. Notice that it is possible that two vertices $x_j, x_{j+1} \in X_L \setminus Y_L$ have corresponding vertices $m',m''$ in the same face (i.e $m'=m''$). Thus, we get that $|M \in F(B)| \geq \left\lceil (b_B-3)/2\right\rceil$.
\end{proof}
Now we are ready to show that the $|Y_L| = O(|M|)$.

\begin{lemma}\label{lem_y3}
Let $c \in \mathcal{N}$ and $c'>0$ be constants such that $$\frac{22c'}{cc'-24c'-2} > 0 \textrm{ \ \ and \ \ } C:= \max{\left\{\frac{22c'}{cc'-24c'-2}, c' \right\}}. \textrm{ \ \ Then \ \ } $$
$$|Y_L \setminus M| \leq C |M| \textrm{ \ \ so \ \ } |Y_L| \leq (C+1)|M|$$
\end{lemma}
\begin{proof}
	 We start with an outline of the proof. Our goal is to show that $|Y_L \setminus M| = O(|M|)$. To this end, we first prove that there are many of edges in the set $E_L$ ($E_L$ was specified in Definition \ref{definition_arrow_graph} on page \pageref{definition_arrow_graph}). Since $E_L$ is large set, the graph $G_1$ contains also many $(Y_L \setminus M)$-$(X_L \setminus Y_L)$-$M$-bunches. In addition using Lemma \ref{lem_many_m_in_big_bunch}, most of them contain a lot of vertices from the optimal solution $M$. Hence, finally we get that $|Y_L|=O(|M|)$.\par
Assume that $|Y_L \setminus M| > c'|M|$ (where $c' >0$). In other case lemma is proved because $C \geq c'$. To estimate the order of the set of edges $E_L$ we will first consider number of edges in sets $E_M$ and $E_S$ in a graph $G_1$. Notice that the graph $G_1$ is planar and sets $X_M$ and $Y_L \setminus M$ are disjoint ($X_M \cap (Y_L \setminus M) = \emptyset$). Hence, from the assumption that $|Y_L \setminus M| > c' |M|$ and Lemma \ref{lem_planar_graph} and Lemma \ref{lem_y1} we get that
 $|E_M| \leq 2(|X_M| + |Y_L \setminus M|) \leq 2(|M| + |Y_L \setminus M|) \leq (2|Y_L \setminus M|(c'+1))/{c'}.$\par
	Notice also that set $E_S$ is empty ($E_S = \emptyset$). Indeed, if there is an edge $e=\{u,v\}$ 
such that $u \in X_S$ and $v \in Y_L \setminus M$ then vertex $u$ would have chosen vertex 
$v \in Y_L \setminus M$, so $u$ would not be in the~set~$X_S$ ($u \notin X_S$). 
	Let $E'_L$ be a subset of $E_L$, where no edge has two endpoints inside $Y_L \setminus M$ set. Thus using planarity we obtain the following inequality
$|E'_L|  \geq |E_L|-6|Y_L \setminus M|.$
 From the definition of the set $Y_L$ we know that each vertex $v \in Y_L$ is adjacent to at least $c$ vertices from $X:=X_M \cup X_S \cup X_L$. Hence, 
\begin{align*}
|E'_L| & \geq c|Y_L  \setminus M|-|E_M|-|E_S| \geq c|Y_L \setminus M| - 0 - \frac{2|Y_L \setminus M|(c'+1)}{c'} \geq |Y_L \setminus M|\cdot \frac{cc'-2c'-2}{c'}
\end{align*}
Observe that there is a bijection from $E'_L$ to a set of $(Y_L \setminus M)$-$(X_L \setminus Y_L)$-$M$-special paths in the graph $G_1$. Thus a graph $G_1$ contains at least $|E'_L|$ special paths. \par
Now we would like to use fact~\ref{lem_bunches_amount} and lemma~\ref{lem_many_m_in_big_bunch} together. To do it we have to ensure that considered bunches contain at least five special paths (assumption of the lemma~\ref{lem_many_m_in_big_bunch}). Recall that if bunch $B$ contains at least five special paths then we say that $B$ is {\it large}. It is obvious from pigeonhole principle, that there are at most $4|\mathcal{B}_1|$ special paths 	which do not belong to large bunches. \par
	Now we will calculate order of the set of $(Y_L \setminus M)$-$(X_L \setminus Y_L)$-$M$-special paths in graph $G_1$ which belong to the set of {\it large bunches}. From Fact \ref{lem_bunches_amount} we know that $|\mathcal{B}_1| \leq 4|Y_L \setminus M|+5|M|$, so at most $16|Y_L \setminus M| + 20|M|$ considered special paths not belong to a set of {\it large bunches}. 
	Let $\mathcal{B}_1^{BIG}$ be a set of special paths which are contained in some {\it large bunch} and such that a internal vertex $x_i$ of each special path $v,x_i,m$ was not added to set $X$ by a vertex $v$ then
\begin{align*}
|\mathcal{B}_1^{BIG}| & \geq |E'_L| - (16|Y_L \setminus M| + 20|M|)
 \geq |Y_L \setminus M| \cdot \frac{cc'-24c'-2}{c'} - 20|M|
\end{align*}
 Using lemma~\ref{lem_many_m_in_big_bunch}  and observing that in calculation of a set $\mathcal{B}_1^{BIG}$ we remove four vertices for each bunch we get that 
\begin{align*}
 |M| & \geq \sum\limits_{B \in \mathcal{B}_1, b_B\geq 5} \left\lceil\frac{b_B-3}{2}\right\rceil \geq \frac{|\mathcal{B}_1^{BIG}|}{2}  \geq \frac{|Y_L \setminus M| \cdot \frac{cc'-24c'-2}{c'} - 20|M| }{2}
\end{align*}
\end{proof}
	Notice that using easily lemmas \ref{lem_y1}, \ref{lem_y2}, and \ref{lem_y3} and assuming proper values for constants $c$ and $c'$ we obtain that $|D_1 \setminus M| =O(|M|)$ and moreover using exactly the same reasoning we could prove following lemma.
\begin{lemma}\label{lem_y_d2}
Let $c$, $c'$ and $C$ be defined as in earlier lemmas. Then
$|D_2 \setminus M| \leq |M| + c|M| +C|M|.$
\end{lemma}
 Thus to prove that our algorithm returns a constant approximation of the MDS problem we have to show that $|D_3|= O(|M|)$. Let us observe that the set $D_3$ contains only vertices which are not dominated by set $D_1$. We divide a set $D_3$ to three pairwise disjoint subsets $D_3 \cap M$, $D_3':=\left\{ v \in (D_3 \setminus M): \exists u \in N_v \setminus M \wedge w(u) = v  \right\}$, and $D_3 \setminus (D'_3 \cup M)$. The orders of the sets $D_3 \cap M$ and $D_3 \setminus (D'_3 \cup M)$ is obvious so  we only have to calculate the size of the set $D'_3$.

\noindent Our last step is to prove that $|D_3 \setminus M| = O(|M|)$. 
\begin{defi}\label{definition_arrow_graph_2}
	Let graph  $G_2=(V_{2},E_{2})$ be a subgraph of $G=(V_G,E_G)$ constructed in the following way:
\vspace{-5pt}
\begin{itemize} \itemsep=-2pt \topsep=-4pt
\item[i)] $V_{2} := \left\{ v \in (D_3 \setminus M): \exists u \in N_v \setminus M \wedge w(u) = v  \right\}$ and $E_{2} = \emptyset$.
\item[ii)] For every vertex $v \in V_2$ add exactly one vertex $u \notin M$. The set of added vertices denote as $U$. Add also edge $\{u, v \}$ to $E_2$.
\item[iii)] Add minimal number of edges (from $E_G$) and nodes (from $D_1 \cup D_2$) such that each vertex $x \in V_2 \setminus U$ has adjacent vertex from the set $D_1 \cup D_2$.
\item[iv)] Add minimal number of edges (from $E_G$) and nodes (from $V_G$) such that each vertex $v \in U$ has adjacent vertex from the set $M$.
\end{itemize}
\end{defi}
\begin{figure}[ht!]
	\centering
	\includegraphics[width=0.50\textwidth]{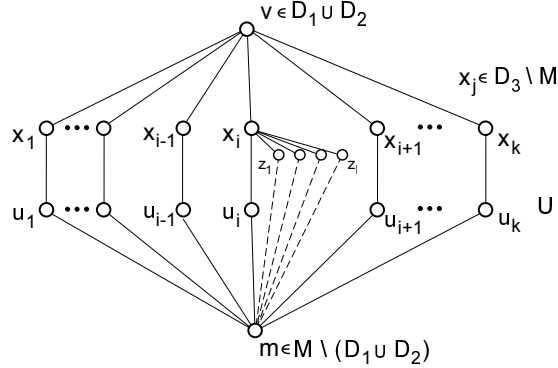}
	\caption{Example of bunch in the graph $G_2$.}
	\label{fig:figure_some_bunch_g3}
\end{figure}

 Notice that $u \in U$ cannot be adjacent to any vertex from a set $D_1$, indeed in other case a vertex $u$ would be dominated by $D_1$ and so it will omit a step \ref{alg_d(v)} of the algorithm.

\begin{fact}\label{lem_bunches_amount_G3}
	Let denote a set of $(D_1 \cup D_2)$-$D_3'$-$U$-$(M \setminus D_1)$-bunches in graph $G_2$ as $\mathcal{B}_2$. Then $|\mathcal{B}_2| \leq (9 + 8c + 8C)|M|.$
\end{fact}
\begin{proof}
	To prove this lemma we need to observe that sets $D_1 \cup D_2$, $D_3'$, $U$, and $(M \setminus D_1)$ in graph $G_2$ are pairwise disjoint. Moreover each vertex $x \in D'_3$ has exactly one adjacent vertex $u \in U$. Thus if $G'$ be a graph constructed from $G$ by contracting each such edge $\{x,u\}$ then we apply this graph in Lemma \ref{lem_bound_number_of_bunches} and obtain that 
$|\mathcal{B}_2| \leq 5(|D_1 \cup D_2| + |M \setminus (D_1 \cup D_2)) + |M| \leq  (9 + 8c + 8C)|M|.$
\end{proof} 

\begin{lemma}\label{lem_many_m_in_big_bunch_second}
	Let $B \in \mathcal{B}_2$ be a bunch such that $B$ contains at least five $(D_1 \cup D_2)$-$D'_3$-$U$-$(M \setminus Y)$-special paths in graph $G_2$ ($b_{B} \geq 5$). Then
	$|M \in F(B)| \geq \left\lceil \frac{b_{B}-4}{2}\right\rceil .$
\end{lemma}
\begin{proof}
The graph $G$ induced by vertices  contained in a region of some bunch $B \in \mathcal{B}_2$ ($R[B]$) looks quite similar like a bunch from a set $\mathcal{B}_1$. 
Using the same reasoning as in the corresponding Lemma \ref{lem_many_m_in_big_bunch} we will obtain that for every vertex $x_i \in F(B) \cap D'_3$ there exists at least one vertex $z$ inside $F(v,x_{i-1},m,x_{i+1},v)$ adjacent to $x_i$ such that at least one following case is satisfied:
\vspace{-5pt} 
\begin{enumerate} \itemsep=-2pt \topsep=-4pt
 	\item[a)] $z \in M$ 
	\item[b)] $\exists m' \in M$ such that $\{m',z\} \in E_G$, $m' \in F(v,x_{i-1},m,x_{i+1},v)$ and $m' \neq m$
\end{enumerate}
	A vertex $v_i$ was added to the set $D_3$ in the step \ref{alg_D3} by vertex $u_i$ thus $$|N^+_{x_i} \cap (V \setminus N^+_D)| \geq |N^+_{m} \cap (V \setminus N^+_D)|.$$ Hence for each an interior vertex $x_i$ there exist adjacent vertices $z_1, \ldots, z_l$ which are not dominated by any vertex $v \in D_1 \cup D_2$. 
Suppose that $z_1, \ldots, z_l \notin M$ and case b) is not satisfied then vertices $z_1, \ldots, z_l$ must be adjacent with single vertex $m \in M \setminus (D_1 \cup D_2)$. Let us notice that vertex $m$ is adjacent with at least one vertex $D_1 \cup D_2$ so $d(u_i) = m$. Contradiction that vertex $u_i $ chose $x_i$ in the step \ref{alg_d(v)} of the algorithm. 
	Each internal vertex from a bunch $B$ has a corresponding vertex $m'\in M$ which is contained in one of two surrounding faces. Since two vertices $x_j, x_{j+1} \in D'_3$ could share the same corresponding vertex $m'$ thus we obtain that $|M \in F(B)| \geq \left\lceil (b_B-2)/4\right\rceil$.
\end{proof}

\begin{lemma}\label{lem_d2}
	$|D'_3| \leq (22 + 16c + 16C)|M|.$
\end{lemma}
\begin{proof}
	Let us notice that for any $v \in D_1$ and $w\in U$ there is no edge $\{v,w \}$ in a graph $G_2$. Indeed, in other case a vertex $w$ will be dominated in step \ref{alg_D1} or step \ref{alg_D2} of the algorithm so would not belongs to a set $U$.
Moreover every vertex $u\in U$ must be dominated in $M$ so must be adjacent to some vertex $m \in M \setminus (D_1 \cup D_2)$.  If we denote a set of special paths which are contained in set of {\it large bunches} as $\mathcal{B}_2^{BIG}$ then 
$|\mathcal{B}_2^{BIG}|  \geq |D'_3| - (20 + 16c + 16C) |M|.$
So similarly like in  a lemma~\ref{lem_many_m_in_big_bunch} we get that
\begin{align*}
 |M| & \geq \sum\limits_{B \in \mathcal{B}_2, b_B\geq 3} \left\lceil\frac{b_B-4}{2}\right\rceil \geq \frac{|\mathcal{B}_2^{BIG}|}{2} \geq \frac{|D'_3| -(20 + 16c + 16C) |M| }{2} .
\end{align*}

\end{proof}

\begin{theorem}\label{thm:main_repeat}
	Let $G=(V,E)$ be a planar graph and $D$ be a set returned by the algorithm {\it PortNumberingMds} and $M$ be an optimal solution of the Minimum Dominating Set  for a given graph $G$ then 
$$|D| \leq 636 |M|.$$
\end{theorem}
\begin{proof}
	Let us fix values of constants $c$ and $c'$ in the following way 
$ c:= 29 \textrm{ \ \ and \ \ } c' = 4.8 $.
Then value of $C$ from earliest lemma is equal to $4.8$. We know that the order of the set $D$ returned by our algorithm satisfy a following inequality 
$|D| \leq |D_1 \setminus M| + |D_2 \setminus M| + |D_3 \setminus M| + |M|.$ \par
So using lemmas \ref{lem_y1}, \ref{lem_y_d2}, \ref{lem_y2} and \ref{lem_y3} we obtain that  
	$|D_1 \setminus M|, |D_2 \setminus M| \leq |Y_M| + |Y_S| + |Y_L \setminus M| \leq |M| + c|M| + C|M|.$ Using Lemma \ref{lem_d2}, and simply calculating $|D_3 \cap M|$, and $|D_3 \setminus (D'_3 \cup M)|$ we get that 
\begin{align*}
|D_3 \setminus M| &\leq |D'_3| + |M|  \leq 23|M| + 16c|M| + 16C|M|.
\end{align*}
Thus if we fix constants $c=33$ and $c'=5.6$ then
	$|D| \leq |M| + |D_1 \setminus M| + |D_2 \setminus M| + |D_3 \setminus M| \leq (26 + 18c + 18C)|M \leq 636|M|.$
\end{proof}

\begin{figure}[ht!]
	\centering
	\includegraphics[width=0.85\textwidth]{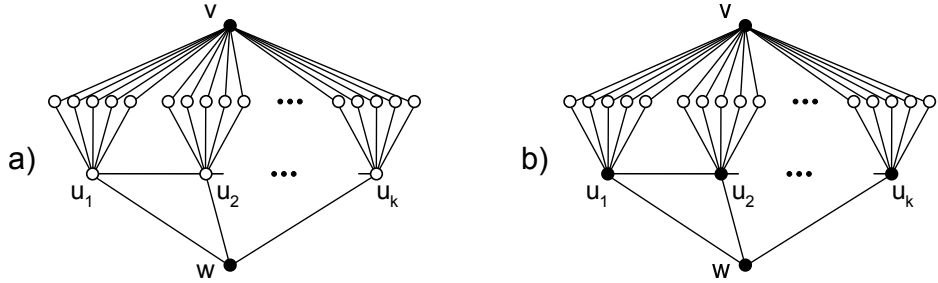}
	\caption{An example of possible results for a) an optimal algorithm b) algorithm in which each vertex chooses one the neigbor of the  largest  degree. The vertices from the resulting Dominating Set are marked(black).}
	\label{fig:example_bad_approximation}
\end{figure}

\section{Conclusion }\label{sec:Conclusion}\parskip 0pt 
In this paper we presented a constant approximation algorithm for the MDS problem in planar graphs. The algorithm is deterministic and strictly local. So nodes do not need any additional information about the structure of the graph and don't have unique identifiers. In our algorithm we use only short messages with  at most $O(\log{ n})$ bits ($\mathcal{CONGEST}$ model). \par
Recently in paper ''Lower Bounds for Local Approximation''\cite{GHS_12a} Mika G\"{o}\"{o}s et al. 
 have shown that for lift-closed bounded degree graphs models {\it PO} and {\it ID} are practically equivalent. In this paper we show that it is true for planar graphs and MDS problem.	We hope that this work will be very helpful as a hint for further comparisons of these models in other classes of graphs. \par
	Moreover the approximation factor is $636$, so there is a large gap to the known lower bound ($5-\epsilon$) from paper \cite{CHW_08a} and approximation factor $130$ from paper \cite{LOW_10tik}. An interesting issue might be a reduction of this gap in a {\it PO} or {\it ID} model.

%
\label{sect:bib}
\bibliographystyle{plain}
\bibliography{easychair}

\begin{thebibliography}{00}




\bibitem{ABI_86}
N. Alon, L. Babai, and A. Itai. A fast and simple randomized parallel
algorithm for the maximal independent set problem. Journal of Algorithms,
7(4):567--583, 1986.

\bibitem{Angluin_80}
D. Angluin. Local and global properties in networks of processors. In Proc. 12th
Annual ACM Symposium on Theory of Computing (STOC, Los Angeles, CA, USA,
April 1980), pages 82--93. ACM Press, New York, NY, USA, 1980.

\bibitem{CV_86}
R. Cole and U. Vishkin. Deterministic Coin Tossing with Applications
to Optimal Parallel List Ranking. Information and Control, 70(1):32--53,
1986.

\bibitem{CHW_08a}
A. Czygrinow, M. Han\'{c}kowiak, and W. Wawrzyniak. Fast distributed approximations in planar graphs. In Proc. 22nd International Symposium on Distributed Computing (DISC, Arcachon, France, September 2008), volume 5218 of Lecture Notes in Computer Science, pages 78--92. Springer, Berlin, Germany, 2008.

\bibitem{CHW_08b}
A. Czygrinow, M. Han\'{c}kowiak, and W. Wawrzyniak. Distributed packing in planar graphs. In the twentieth ACM Symposium on Parallel Algorithms and Architectures, pages 55--61, 2008.

\bibitem{CHKSW11}
A. Czygrinow, M. Han\'{c}kowiak, K. Krzywdzinski, E. Szyma\'{n}ska, and W. Wawrzyniak. Brief announcement: Distributed approximations for the semi-matching problem. In Proc. 25th International Symposium on Distributed Computing (DISC, Rome, Italy, September 2011), volume 6950 of Lecture Notes in Computer Science, pages 200-–201. Springer,Berlin, Germany, 2011.

\bibitem{CHSW12a}
A. Czygrinow, M. Han\'{c}kowiak, E. Szymanska, and W. Wawrzyniak. Distributed 2-approximation algorithm for the semi-matching problem. In Proc. 26th International Symposium on Distributed Computing (DISC, Salvador, Brazil,  October 2012), volume 7611 of Lecture Notes in Computer Science, pages 210-–222. Springer, Berlin, Germany, 2012.

\bibitem{GHS_12a}
M. G\"{o}\"{o}s, J. Hirvonen, and J. Suomela , Lower Bounds for Local Approximation, In Proc. 31st Annual ACM Symposium on Principles of
Distributed Computing (PODC, Madeira, Portugal, July 2012), pages 175--184. ACM Press, New York, NY, USA, 2012.

\bibitem{II_86}
A. Israel and A. Itai. A fast and simple randomized parallel algorithm for
maximal matching. Information Processing Letters, 22(2):77--80, 1986.

\bibitem{KMW_06}
F. Kuhn, T. Moscibroda, and R. Wattenhofer. The price of being
near-sighted. In Proc. 17th Annual ACM-SIAM Symposium on Discrete Algorithms
(SODA, Miami, FL, USA, January 2006), pages 980--989. ACM Press, New York,
NY, USA, 2006.

\bibitem{KW_05}
F. Kuhn and R. Wattenhofer. Constant-time distributed dominating set
approximation. Distributed Computing, 17(4):303--310, 2005.

\bibitem{LOW_10tik}
C. Lenzen, Y. A. Oswald, and R. Wattenhofer. What can be
approximated locally? TIK Report 331, ETH Zurich, Computer Engineering and
Networks Laboratory, November 2010.

\bibitem{LW_08}
C. Lenzen and R. Wattenhofer. Leveraging
Linial's locality limit. In Proc. 22nd Symposium on
Distributed Computing (DISC 2008), volume 5218 of
LNCS, pages 394--407. Springer, Berlin, 2008.

\bibitem{L_92a}
N. Linial. Locality in Distributed Graph Algorithms. SIAM Journal on
Computing, 21(1):193--201, 1992.

\bibitem{Luby_86}
M. Luby. A Simple Parallel Algorithm for the Maximal Independent Set
Problem. SIAM Journal on Computing, 15(4):1036--1055, 1986

\bibitem{NS_95}
M. Naor and L. Stockmeyer. What Can Be Computed Locally? SIAM
Journal on Computing, 24(6):1259--1277, 1995.

\bibitem{Peleg_00}
D. Peleg. Distributed Computing: A Locality-Sensitive Approach. Society
for Industrial and Applied Mathematics, Philadelphia, PA, USA, 2000.

\bibitem{Suomela_13}
J. Suomela, Survey of local algorithms. ACM Computing Surveys (to appear),
http://www.cs.helsinki.fi/local-survey/




\end{thebibliography}

\vfill\eject

\end{document}